\documentclass[DIV9,a4paper,final]{scrartcl}

\usepackage[utf8]{inputenc}
\usepackage[T1]{fontenc}

\usepackage{amssymb}
\usepackage{booktabs}
\usepackage{enumitem}
\usepackage{float}
\usepackage{lmodern}
\usepackage{microtype}
\usepackage{nccmath}
\usepackage{tikz}

\usetikzlibrary{automata}

\newtheorem{theorem}{Theorem}
\newtheorem{lemma}[theorem]{Lemma}

\newtheorem{proposition}[theorem]{Proposition}
\newcommand{\qed}{\hspace*{\fill}\ensuremath{\Box}}
\newenvironment{proof}{\pagebreak[3]\noindent\textbf{Proof.}}{\qed\pagebreak[3]\medskip}

\setlist{itemsep=1pt,parsep=0pt,topsep=2pt}

\renewcommand{\geq}{\geqslant}
\renewcommand{\leq}{\leqslant}
\renewcommand{\ge}{\geq}
\renewcommand{\le}{\leq}

\newcommand{\autA}{\mathcal{A}}
\newcommand{\autF}{\mathcal{F}}
\newcommand{\fullaut}[3]{\autF({#1}, {#2}, {#3})}
\newcommand{\bigO}[1]{\ensuremath{\mathcal{O}({#1})}}

\newcommand{\union}{\mathbin{\cup}}
\newcommand{\bigunion}{\mathop{\bigcup}}
\newcommand{\intersect}{\mathbin{\cap}}

\newcommand{\NFA}{\ensuremath{\textrm{NFA}}}
\newcommand{\NBA}{\ensuremath{\textrm{BA}}}
\newcommand{\LNFA}{\ensuremath{L_{\NFA}}}
\newcommand{\LBA}{\ensuremath{L_{\NBA}}}

\newcommand{\abs}[1]{\ensuremath\left|#1\right|}
\newcommand{\set}[2]{\ensuremath{\left\{#1 \mid #2\right\}}}
\newcommand{\os}[1]{\ensuremath{\left\{#1\right\}}}
\newcommand{\gR}{\ensuremath{\mathcal{R}}}
\newcommand{\gL}{\ensuremath{\mathcal{L}}}
\newcommand{\gH}{\ensuremath{\mathcal{H}}}
\newcommand{\gJ}{\ensuremath{\mathcal{J}}}
\newcommand{\Req}{\mathrel{\gR}}
\newcommand{\Leq}{\mathrel{\gL}}
\newcommand{\Heq}{\mathrel{\gH}}
\newcommand{\Jeq}{\mathrel{\gJ}}

\newcommand{\Synt}[1]{{{#1}^+} / {\equiv_L}}
\newcommand{\ie}{i.e.,~}
\newcommand{\eg}{e.g.~}
\newcommand{\citenew}{[\textbf{new}]}
\newcommand{\ms}{\hspace*{0.5pt}}

\let\oldpar\paragraph
\renewcommand{\paragraph}[1]{\oldpar*{\bf #1}}

\title{Operations on Weakly Recognizing Morphisms}
\author{Lukas Fleischer \and Manfred Kuf\-leitner}
\date{FMI, University of Stuttgart\thanks{This work was supported by the DFG grants DI 435/5-2 and \mbox{KU 2716/1-1}.}\\[.1mm]
  \normalsize\texttt{\{fleischer,kufleitner\}@fmi.uni-stuttgart.de}}

\begin{document}

\maketitle

\vspace{-8mm}
\begin{abstract}
  \noindent
  {\sffamily\normalsize\bfseries{Abstract.}} \ 
  Weakly recognizing morphisms from free semigroups onto finite semigroups are
  a classical way for defining the class of $\omega$-regular languages, \ie a
  set of infinite words is weakly recognizable by such a morphism if and only
  if it is accepted by some B\"uchi automaton. We consider the descriptional
  complexity of various constructions for weakly recognizing morphisms. This
  includes the conversion from and to B\"uchi automata, the conversion into
  strongly recognizing morphisms, and  complementation. For some problems, we
  are able to give more precise bounds in the case of binary alphabets or
  simple semigroups.
\end{abstract}

\section{Introduction}

B\"uchi automata define the class of $\omega$-regular languages.  They were
introduced by B\"uchi for deciding the monadic second-order theory of
$(\mathbb{N},<)$~\cite{Buc60}. Since then, $\omega$-regular languages have
become an important tool in formal verification, and many other automata models
for this language class have been considered; see
\eg\cite{pp04:short,tho90handbook:short}. Each automaton model has its merits
and its disadvantages. Recently, the authors have shown that recognizing
morphisms have many nice algorithmic
properties~\cite{FleischerKufleitner2015fsttcs:short}. Such morphisms come in
two different flavors. Strongly recognizing morphisms admit efficient
minimization and complementation, whereas weakly recognizing morphisms can be
exponentially more succinct (but there is no minimal weak recognizer and there
is no efficient complementation). The situation is similar to the behavior of
deterministic and nondeterministic finite automata. The major difference to
both nondeterministic finite automata and B\"uchi automata is that there is an
efficient inclusion test for weakly recognizing
morphisms~\cite{FleischerKufleitner2015fsttcs:short}. Every strongly recognizing morphism is also weakly recognizing, but the converse is false.

In this paper, we consider the descriptional complexity of various operations
on weakly recognizing morphisms and conversions involving nondeterministic
Büchi automata (\NBA) and strongly recognizing morphisms.
In each case, we give asymptotically tight bounds. For the conversion of a
$\NBA$ into a weakly recognizing morphism, we give a lower bound which matches
the naive upper bound. Our results are summarized in Table~\ref{tab:fragments}.

There are some similarities between recognizing morphisms over finite and over
infinite words. Strong recognition is the natural counterpart to recognition
for finite words. Nevertheless, in order to prove lower bounds for the
conversion of B\"uchi automata to weakly recognizing morphisms, we first show
that bounds for converting nondeterministic finite automata to recognizing
morphisms over finite words (with some limitations) also hold for the
conversion of Büchi automata to weakly recognizing morphisms. We then use
techniques of Sakoda and Sipser~\cite{SakodaSipser78stoc:short} and of
Yan~\cite{Yan08lmcs} to obtain tight bounds for the conversion of
nondeterministic finite automata to recognizing morphisms. This step is similar
to the work of Holzer and König~\cite{HolzerKoenig2004tcs}. To the best of our
knowledge, our lower bound over finite words for the conversion of an NFA into
a recognizing morphism is also a new result.

\renewcommand{\arraystretch}{1.2}
\begin{table}[t]\label{tab:fragments}
  \centering\small
  \begin{tabular}{l@{\hspace*{5mm}}r@{\hspace*{2.5mm}}r}
    \toprule
    \textbf{Operation} & \textbf{Lower bound} & \ \ \textbf{Upper bound} \\
    \midrule
    $\NBA$ to weak recognition & $2^{n^2}$~\citenew & $2^{n^2}$~\cite{pecu86stacs:short} \\
    $\NBA$ to weak recognition, binary alphabet
     & $2^{{(n-1)}^2/4}$~\citenew & $2^{n^2}$~\cite{pecu86stacs:short} \\
    Weak recognition to $\NBA$ & \hspace*{-4mm}$(n-3)(n+1)/32$~\citenew & $n(n + 1)$~\cite{pecu86stacs:short} \\
    Weak recognition to strong recognition & $n \ms 2^{n-1}$~\citenew & $2^{n^2}$~\cite{pp04:short} \\
    Complementation of weak recognition & $n \ms 2^{n-1}$~\citenew & $2^{n^2}$~\cite{pp04:short} \\
    Complementation for simple semigroups & $n \ms 2^{n-1}$~\citenew & $n \ms 2^n$~\citenew \\
    \bottomrule
  \end{tabular}
  \smallskip
  \caption{Bounds for the descriptional complexity of various operations.}
\end{table}

\section{Preliminaries}

This section gives a brief overview of some basic definitions from the fields of
formal languages, finite automata and semigroup theory. We refer to~\cite{pp04:short,pin86:short} for more detailed introductions.

\paragraph{Words.}

Let $A$ be a finite \emph{alphabet}. The elements of $A$ are called
\emph{letters}.
A \emph{finite word} is a sequence $a_1 a_2 \cdots a_n$ of letters of $A$ and
an \emph{infinite word} is an infinite sequence $a_1 a_2 \cdots$. The empty
word is denoted by $\varepsilon$.
Given an infinite word $\alpha = a_1 a_2 \cdots$, we let $\inf(\alpha) \subseteq A$ denote the set of letters in $\alpha$ which occur infinitely often.

Let $K$ be a set of finite words and let $L$ be a set of infinite words. We set
$KL = \set{u\alpha}{u \in K, \alpha \in L}$, $K^n = \set{u_1 u_2 \cdots
u_n}{u_i \in K}$, $K^+ = \bigunion_{n \ge 1} K^n$ and $K^* = K^+ \union
\os{\varepsilon}$.
Moreover, if $\varepsilon \not\in K$ we define the \emph{infinite iteration}
$K^\omega = \set{u_1 u_2 \cdots}{u_i \in K}$. A natural extension to $K
\subseteq A^*$ is $K^\omega = {(K \setminus \os{\varepsilon})}^\omega \union
\os{\varepsilon}$.

\paragraph{Automata.}

A \emph{finite automaton} is a 5-tuple $\autA = (Q, A, \delta, I, F)$ where
$Q$ is a finite set of \emph{states} and $A$ is a finite alphabet. The
\emph{transition relation} $\delta$ is a subset of $Q \times A \times Q$ and
its elements are called \emph{transitions}. The sets $I$ and $F$ are subsets of
$Q$ and are called \emph{initial states} and \emph{final states}, respectively.

A \emph{finite run} of a word $a_1 a_2 \cdots a_n$ on $\autA$ is a sequence
$q_0 a_1 q_1 a_1 \cdots q_{n-1} a_n q_n$ such that $q_0 \in I$ and $(q_i,
a_{i+1}, q_{i+1}) \in \delta$ for all $i \in \os{0, \dots, n-1}$. The run is
said to \emph{start} in $q_0$ and \emph{end} in $q_n$. The word $a_1 a_2 \cdots
a_n$ is the \emph{label} of the run.
A finite run is called \emph{accepting} if it ends in a final state.
A finite word $u$ is said to be \emph{accepted by $\autA$} if there exists an
accepting finite run of $u$ on $\autA$ and the language \emph{accepted by
$\autA$} is the set of all finite words over $A^*$ accepted by $\autA$. It is
denoted by $\LNFA(\autA)$.

Analogously, an \emph{infinite run} of a word $a_1 a_2 \cdots$ on $\autA$ is an
infinite sequence $q_0 a_1 q_1 a_1 \cdots$ such that $q_0 \in I$ and $(q_i,
a_{i+1}, q_{i+1}) \in \delta$ for all $i \ge 0$.
It is called \emph{accepting} if $\inf(q_0 q_1 q_2 \cdots) \intersect F \ne
\emptyset$.
An infinite word $\alpha$ is said to be \emph{Büchi-accepted by $\autA$} if
there exists an accepting infinite run of $\alpha$ on $\autA$.
The language \emph{Büchi-accepted by $\autA$} is the set of all infinite
words Büchi-accepted by $\autA$ and it is denoted by $\LBA(\autA)$.

We use the term \emph{run} for both finite and infinite runs if the reference
is clear from the context.
A language $L \subseteq A^*$ (resp.\ $L \subseteq A^\omega$) is \emph{regular}
(resp.\ \emph{$\omega$-regular}) if it is accepted (resp.\ Büchi-accepted) by
some finite automaton.

\paragraph{Finite semigroups.}

A \emph{semigroup morphism} is a mapping $h \colon S \to T$ between two (not
necessarily finite) semigroups $S$ and $T$ such that $h(s)h(t) = h(st)$ for all
$s, t \in S$. Since we do not consider morphisms of other objects, we use the
term \emph{morphism} synonymously.
A \emph{subsemigroup} of a semigroup $S$ is a subset that is closed under
multiplication.
We say that a semigroup $T$ \emph{divides} a semigroup $S$ if there exists a
surjective morphism from a subsemigroup of $S$ onto $T$.

\emph{Green's relations} are an important tool in the study of semigroups.
For the remainder of this subsection, let $S$ be a finite semigroup.
We let $S^1$ denote the monoid that is obtained by adding a new neutral element
$1$ to $S$.
For $s, t \in S$ let
\begin{ceqn}
  \begin{align*}
    s & \Req t \text{~if there exist~} q, q' \in S^1 \text{~such that~} sq = t \text{~and~} tq' = s, \\
    s & \Leq t \text{~if there exist~} p, p' \in S^1 \text{~such that~} ps = t \text{~and~} p't = s, \\
    s & \Jeq t \text{~if there exist~} p, q, p', q' \in S^1 \text{~such that~} psq = t \text{~and~} p'tq' = s, \\
    s & \Heq t \text{~if~} s \Req t \text{~and~} s \Leq t.
  \end{align*}
\end{ceqn}
These relations are equivalence relations. The equivalence classes of $\Req$
(resp.~$\Leq$, $\Jeq$, $\Heq$) are called \emph{$\gR$-classes}
(resp.~\emph{$\gL$-classes}, \emph{$\gJ$-classes}, \emph{$\gH$-classes}).
For $s \in S$, we denote the $\gR$-class (resp.~$\gL$-class) of $s$ by $R_s$
(resp.~$L_s$) and we let ${S} / {\gR} = \set{R_s}{s \in S}$ as well as ${S} /
{\gL} = \set{L_s}{s \in S}$.

A semigroup is called \emph{$\gJ$-trivial} if each of its $\gJ$-classes
contains exactly one element. A semigroup is called \emph{simple} if it
consists of a single $\gJ$-class. In a finite simple semigroup, the relations
$s \Req st \Leq t$ hold for all $s, t \in S$. Moreover, each $\gH$-class forms
a group and all such groups are isomorphic~\cite{pin86:short}. We will also utilize
the following lemma:
\begin{lemma}\label{lem:inverse}
  Let $S$ be a finite simple semigroup and let $x, y, z \in S$ such that $y
  \Req z$.  Then $xy = xz$ implies $y = z$.
\end{lemma}
\begin{proof}
  Suppose that $xy = xz$.
  Since $S$ is simple, we have $y \Leq xy$ and thus, there exists an element $p
  \in S^1$ such that $pxy = y$. Since $y \Req z$, there exists an element $q
  \in S^1$ with $yq = z$. It follows that $y = pxy = pxz = pxyq = yq = z$.
\end{proof}

\paragraph{Recognition by morphisms.}

Let $h \colon A^+ \to S$ be a morphism to a finite semigroup $S$.
A pair $(s,e)$ of elements of $S$ is a \emph{linked pair} if $se = s$ and $e^2
= e$.
For $s \in S$, we set ${[s]}_h = h^{-1}(s)$ and if $h$ is understood from the
context, we may skip the reference to the morphism in the subscript.
A language $L \subseteq A^+$ is \emph{recognized} by a morphism $h: A^+ \to S$
if $L$ is a union of sets $[s_i]$ with $s_i \in S$.
A language $L \subseteq A^\omega$ is \emph{weakly recognized} by a morphism
$h: A^+ \to S$ if it is a union of sets $[s_i]{[e_i]}^\omega$ where $(s_i,
e_i)$ are linked pairs of $S$.
A language $L \subseteq A^\omega$ is \emph{strongly recognized} by a morphism
$h: A^+ \to S$ if $[s]{[t]}^\omega \intersect L \ne \emptyset$ implies
$[s]{[t]}^\omega \subseteq L$ for all $s, t \in S$.
It is easy to see that strong recognition implies weak recognition,
see~\eg\cite[Theorem~2.2]{pp04:short}.
Moreover, if a morphism strongly recognizes $L$, it also strongly recognizes
its complement $A^\omega \setminus L$.
By extension, we also say that a semigroup $S$ recognizes (resp.\ weakly
recognizes, strongly recognizes) a language $L$ if there exists a morphism $h
\colon A^+ \to L$ that recognizes (resp.\ weakly recognizes, strongly
recognizes) $L$.

For a language $L \subseteq A^+ \union A^\omega$, we have $u \equiv_L v$ if and
only if
\begin{ceqn}
  \begin{align*}
    (xuy)z^\omega \in L & \Leftrightarrow (xvy)z^\omega \in L \text{~and~} \\
    z{(xuy)}^\omega \in L & \Leftrightarrow z{(xvy)}^\omega \in L
  \end{align*}
\end{ceqn}
for all finite words $x, y, z \in A^*$. Keep in mind that $\varepsilon^\omega =
\varepsilon$. The relation $\equiv_L$ was introduced by Arnold~\cite{arn85}; it
is called the \emph{syntactic congruence} of $L$.
The congruence classes of $\equiv_L$ form the so-called \emph{syntactic
semigroup} $\Synt{A}$ and the \emph{syntactic morphism} $h_L \colon A^+ \to
\Synt{A}$ is the natural quotient map.
If $L \subseteq A^*$ (resp.\ $L \subseteq A^\omega$) is regular (resp.\
$\omega$-regular), the syntactic semigroup of $L$ is finite and $h_L$
recognizes (resp.\ strongly recognizes) the language $L$;
see~\cite{arn85,pp04:short}.

\section{Lower Bound Techniques}

\subsection{Proving Lower Bounds for Weakly Recognizing Morphisms}

We first consider the general problem of proving lower bounds for the size of
weakly recognizing semigroups for a given language $L$. In the case of
recognizing morphisms over finite words and in the case of strongly recognizing
morphisms, this is easy since one only needs to compute the syntactic
semigroup, which immediately yields a tight lower bound. On the contrary,
weakly recognizing morphisms do not admit minimal objects. However, it turns
out that one can still use a relaxed version of Arnold's syntactic congruence.

We first prove a combinatorial lemma and then give the main result of this
section.

\begin{lemma}\label{lem:technical}
  Let $u, v \in A^+$ and let $(s, e)$ be a linked pair. Then $uv^\omega$ is
  contained in $[s]{[e]}^\omega$ if and only if there exists a factorization $v
  = v_1 v_2$ and powers $k, \ell \ge 0$ such that $\ell$ is odd, $h(u v^k v_1)
  = s$ and $h(v_2 v^\ell v_1) = e$.
\end{lemma}
\begin{proof}
  Let $v = a_1 a_2 \cdots a_n$ with $n \ge 1$ and $a_i \in A$. If $uv^\omega$
  is contained in $[s]{[e]}^\omega$, there exists a factorization $uv^\omega =
  u' v_1' v_2' \cdots$ such that $h(u') = s$ and $h(v_i') = e$ for all $i \ge
  1$.
  Since $u$ and $v$ are finite words, there exist indices $j > i \ge 1$, powers
  $k, \ell \ge 1$ and a position $m \in \os{1, \dots, n}$ such that $u' v_1'
  v_2' \cdots v_{i-1}' = u v^k a_1 a_2 \cdots a_m$ and $v_i' v_{i+1}' \cdots
  v_j' = a_{m+1} a_{m+2} \cdots a_n v^\ell a_1 a_2 \cdots a_m$.
  We set $v_1 = a_1 a_2 \cdots a_m$ and $v_2 = a_{m+1} a_{m+2} \cdots a_n$.
  Then $v_1 v_2 = v$,
  \begin{ceqn}
    \begin{alignat*}{4}
      h(u v^k v_1) &= h(u v^k a_1 a_2 \cdots a_m) & &= h(u' v_1' v_2' \cdots v_{i-1}') & &= se^{i-1} & &= s, \\
      h(v_2 v^\ell v_1) &= h(a_{m+1} a_{m+2} \cdots a_n v^\ell a_1 a_2 \cdots a_m) & &= h(v_i' v_{i+1}' \cdots v_j') & &= e^{j-i+1} & &= e.
    \end{alignat*}
  \end{ceqn}
  If $\ell$ is even, we can replace $\ell$ by $2\ell + 1$ since $h(v_2
  v^{2\ell+1} v_1) = h(v_2 v^\ell v_1 v_2 v^\ell v_1) = e^2 = e$.
  The converse implication is trivial.
\end{proof}

\begin{theorem}\label{thm:lower}
  Let $L \subseteq A^\omega$ be a language weakly recognized by some morphism
  $h \colon A^+ \to S$ and let $u, v, z \in A^+$ and $x, y \in A^*$ be words
  such that one of the following two properties holds:
  \begin{enumerate}
    \item\label{enum:loweraaa}  $xuyz^\omega \in L$ and $xvyz^\omega \not\in L$
    \item\label{enum:lowerbbb} $x{(uy)}^\omega \in L$ and $x{(uyvy)}^\omega \not\in L$ and $x{(vyuy)}^\omega \not\in L$.
  \end{enumerate}
  Then $h(u) \ne h(v)$.
\end{theorem}

\begin{proof}
  We consider finite words $u, v \in A^+$ such that $h(u) = h(v)$ and show that
  in this case, neither of the properties can hold.

  If the first property holds, there exists a linked pair $(s, e)$ such that
  $xuyz^\omega \in [s]{[e]}^\omega \subseteq L$. Thus, by
  Lemma~\ref{lem:technical}, we have $h(xuy z^k z_1) = s$ and $h(z_2 z^\ell
  z_1) = e$ for some factorization $z = z_1 z_2$ and powers $k, \ell \ge 0$.
  Now, since $h(xvy z^k z_1) = h(xuy z^k z_1) = s$, we obtain $xvyz^\omega \in
  [s]{[e]}^\omega \subseteq L$, a contradiction.

  If the second property holds, there exists a linked pair $(s, e)$ of $S$ such
  that $xw^\omega \in [s]{[e]}^\omega \subseteq L$ where $w = uy$. Thus, by
  Lemma~\ref{lem:technical}, we have $h(x w^k w_1) = s$ and $h(w_2 w^\ell w_1)
  = e$ for some factorization $w = w_1 w_2$, some power $k \ge 0$ and some odd
  power $\ell \ge 0$. Since $\ell$ is odd $(\ell-1)/2$ is an integer and we
  have $h(w_2 {(vyuy)}^{(\ell-1) / 2} vy w_1) = h(w_2 {(uy)}^\ell w_1) = e$.
  Now, if $k$ is odd as well, we obtain $h(x {(vyuy)}^{(k-1)/2} vy w_1) = h(x
  {(uy)}^k w_1) = s$ and therefore, $x{(vyuy)}^\omega \in L$. Equivalently, if
  $k$ is even, we have $h(x {(uyvy)}^{k/2} w_1) = h(x {(uy)}^k w_1) = s$ and
  hence, $x{(uyvy)}^\omega \in L$. Both cases contradict
  Property~\ref{enum:lowerbbb} above.
\end{proof}

The next proposition is another simple, yet useful, tool for proving lower
bounds. It allows to transfer bounds from the setting of finite words to
infinite words.

\begin{proposition}\label{prop:transfer}
  Let $\autA = (Q, A, \delta, I, F)$ and let $a \in A$ be a letter such
  that for all $q \in Q$ and $q_f \in F$, we have $(q, a, q_f) \in \delta$ if
  and only if $q = q_f$. Let $K = \LBA(\autA)$ and let $L = \LNFA(\autA)$.
  Then each semigroup weakly recognizing $K$ has at least $\abs{\Synt{A}}$
  elements.
\end{proposition}
\begin{proof}
  Let $h \colon A^+ \to S$ be a morphism weakly recognizing $K$ and consider
  two words $u, v \in A^+$ such that $u \not\equiv_L v$. Then, without loss of
  generality, there exist $x, y \in A^*$ such that $xuy \in L$ and $xvy \not\in
  L$. This implies $xuya^\omega \in K$ since $(q_f, a, q_f) \in \delta$
  for all $q_f \in F$.  Equivalently, because of $(q, a, q_f) \not\in \delta$
  for all $q \in Q \setminus F$ and $q_f \in F$, we have $xvya^\omega \not\in
  K$. By Theorem~\ref{thm:lower}, this yields $h(u) \ne h(v)$.
\end{proof}

\subsection{The Full Automata Technique}

The \emph{full automata technique} is a useful tool for proving lower bounds
for the conversion of automata to other objects. It was introduced by
Yan~\cite{Yan08lmcs} who attributes it to Sakoda and
Sipser~\cite{SakodaSipser78stoc:short}. The technique works for both accepted
and Büchi-accepted languages.
However, we will prove the main result of this section only for the setting of
finite words and use Proposition~\ref{prop:transfer} to obtain analogous
results for infinite words.

Let $Q$ be a finite set and let $I, F$ be subsets of $Q$. The \emph{full
automaton} $\fullaut{Q}{I}{F}$ is the finite automaton $(Q, B, \Delta, I, F)$
defined by $B = 2^{Q^2}$ and by the transition relation $\Delta = \set{(p, T,
q) \in Q \times B \times Q}{(p, q) \in T}$.

\begin{theorem}
  Let $\autA = (Q, A, \delta, I, F)$ be a finite automaton and let
  $\fullaut{Q}{I}{F} = (Q, B, \Delta, I, F)$ be the corresponding full
  automaton.
  Then the syntactic semigroup of $\LNFA(\autA)$ divides the syntactic
  semigroup of $\LNFA(\fullaut{Q}{I}{F})$.
\end{theorem}
\begin{proof}
  We first define a morphism $\pi \colon A^+ \to B^+$ by $\pi(a) = \set{(p,
  q)}{(p, a, q) \in \delta}$. Let $K = \LNFA(\fullaut{Q}{I}{F})$ and let $L =
  \LNFA(\autA)$.
  It suffices to show that $\pi(u) \equiv_K \pi(v)$ implies $u \equiv_L v$.
  Thus, consider $u, v \in A^+$ such that $\pi(u) \equiv_K \pi(v)$. In
  particular, for all $x, y \in A^*$, we have $\pi(xuy) \in K$ if and only if
  $\pi(xvy) \in K$.
  By the definition of $\pi$, we have $\pi(w) \in K$ if and only if $w \in L$
  for all $w \in A^+$. Using the equivalence from above, this yields $xuy \in
  L$ if and only if $xvy \in L$ for all $x, y \in A^*$, thereby proving that
  $u \equiv_L v$.
\end{proof}

\section{From Automata to Weakly Recognizing Morphisms}
\label{sec:aut-hom}

The standard construction for converting a finite automaton $\autA$ to a
recognizing morphism is the so-called \emph{transition semigroup} of $\autA$.
For a given word $u \in A^+$, it encodes for each pair $(p, q)$ of states
whether there is a run of $u$ on $\autA$ starting in $p$ and ending in $q$.
Thus, for a finite automaton with $n$ states the transition semigroup has
$2^{n^2}$ elements. For details on the construction, we refer
to~\cite{pp04:short,pin86:short}. We show that this construction is optimal.

\begin{theorem}
  Let $\autA$ be a finite automaton with $n$ states. Then there exists a
  semigroup recognizing $\LNFA(\autA)$ (resp.\ weakly recognizing
  $\LBA(\autA)$) which has at most $2^{n^2}$ elements and this bound is tight.
\end{theorem}
\begin{proof}
  Each language that is accepted (resp.\ Büchi-accepted) by $\autA$ is
  recognized (resp.\ weakly recognized) by the transition semigroup of $\autA$
  which has size $2^{n^2}$.

  To show that this is optimal, we consider the full automaton
  $\fullaut{N}{N}{N} = (N, B, \Delta, N, N)$ where $N = \os{1, \dots, n}$ and
  let $L = \LNFA(\fullaut{N}{N}{N})$.
  For two different letters $X, Y \in B$ we may assume, without loss of
  generality, that there exist $p, q \in N$ such that $(p, q) \in X \setminus
  Y$. With $P = \os{(p, p)}$ and $Q = \os{(q, q)}$, we then have $PXQ \in L$
  and $PYQ \not\in L$. Thus, $X \not\equiv_L Y$.
  This shows that $\Synt{B}$ has at least $\abs{B} = 2^{n^2}$ elements.

  Noting that the transitions labeled by the letter $\set{(q, q)}{q \in N}$
  form self-loops at each state, the Büchi case immediately follows by
  Proposition~\ref{prop:transfer}.
\end{proof}

The proof of the optimality result requires a large alphabet that grows
super-exponentially in the number of states of the automaton.
A natural restriction is considering automata over fixed-size alphabets.

By a result of Chrobak~\cite{Chrobak86}, the size of the syntactic semigroup of
an unary language accepted by a finite automaton of size $n$ is in
$2^{\bigO{\sqrt{n \log n}}}$ (note that since unary languages are commutative,
the syntactic monoid is isomorphic to the minimal deterministic automaton).
Over infinite words, the unary case is uninteresting since the only language
over the alphabet $A = \os{a}$ is $\os{a^\omega}$.

For binary alphabets, a lower bound can be obtained by combining the full
automata technique with a result from the study of semigroups of binary
relations~\cite[Proposition~6]{KimRoush78jmp}. In order to keep the paper
self-contained, we present a proof that is adapted to finite automata and does
not require any knowledge of binary relations.
\begin{theorem}
  Let $A = \os{a, b}$ and let $n$ be an odd natural number. There exists a
  language $L \subseteq A^+$ (resp.\ $L \subseteq A^\omega$) and a finite
  automaton with $n$ states accepting (resp.\ Büchi-accepting) $L$, such that
  each semigroup recognizing (resp.\ weakly recognizing) $L$ has at least
  $2^{{(n-1)}^2/4}$ elements.
\end{theorem}
\begin{proof}
  We first analyze the case of finite words.
  Let $m = (n-1)/2$ and let $M = \os{1, \dots, m}$. We consider the automaton
  $\autA$ depicted below and let $L = \LNFA(\autA)$.
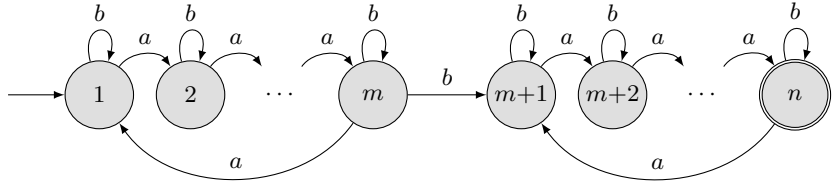
  \begin{figure}[H]
    \centering
    \begingroup
    \colorlet{Alightgray}{gray!25}
    \setlength{\medmuskip}{0mu}
    \begin{tikzpicture}[node distance=1.2cm,>=latex,bend angle=55,initial text=,every state/.style={minimum size=9mm,inner sep=0pt},every loop/.style={in=75,out=105,looseness=6}]
      \tikzstyle{every node}=[font=\footnotesize]
      \node[state,fill=Alightgray]           (1)               {$1$};
      \node[state,fill=Alightgray]           (2) [right of=1] {$2$};
      \node[state,draw=none] (3) [right of=2] {$\cdots$};
      \node[state,fill=Alightgray]           (4) [right of=3] {$m$};
      \node[state,fill=Alightgray]           (5) [xshift=0.75cm,right of=4] {$m+1$};
      \node[state,fill=Alightgray]           (6) [right of=5] {$m+2$};
      \node[state,draw=none] (7) [right of=6] {$\cdots$};
      \node[state,accepting,fill=Alightgray] (8) [right of=7] {$n$};
      \coordinate[left of=1,xshift=0cm] (s);

      \path[->] (s) edge [above] node {} (1);
      \path[->] (1) [bend left] edge [above] node {$a$} (2);
      \path[->] (2) [bend left] edge [above] node {$a$} (3);
      \path[->] (3) [bend left] edge [above] node {$a$} (4);
      \path[->] (4) [bend left] edge [above] node {$a$} (1);
      \path[->] (5) [bend left] edge [above] node {$a$} (6);
      \path[->] (6) [bend left] edge [above] node {$a$} (7);
      \path[->] (7) [bend left] edge [above] node {$a$} (8);
      \path[->] (8) [bend left] edge [above] node {$a$} (5);

      \path[->] (1) [loop above] edge [above] node {$b$} (1);
      \path[->] (2) [loop above] edge [above] node {$b$} (2);
      \path[->] (4) [loop above] edge [above] node {$b$} (4);
      \path[->] (4) edge [above] node {$b$} (5);
      \path[->] (5) [loop above] edge [above] node {$b$} (5);
      \path[->] (6) [loop above] edge [above] node {$b$} (6);
      \path[->] (8) [loop above] edge [above] node {$b$} (8);
    \end{tikzpicture}
    \endgroup
  \end{figure}
  \noindent
  For $1 \le i,j \le m$ we first define $p_{i,j} = (m+j-i)m - i$ and $q_{i,j} =
  (m+i-j+2)m + i$.
  Furthermore, we set $u_{i,j} = a^{p_{i,j}}ba^{q_{i,j}}$.
  We claim that for each $i, j$ there exists a path from state $k$ to $\ell$
  labeled by $u_{i,j}$ if and only if $(k, \ell) = (i, j+m)$ or $k = \ell$.

  The two $a$-cycles have length $m$ and $m+1$, respectively. Since for each
  pair $(i, j)$ we have $p_{i,j} + q_{i,j} = 2m(m + 1)$ and since one can
  always stay in the same state when reading the letter $b$, there clearly
  exists a path from each state to itself labeled by $u_{i,j}$.
  Now, fix some $(i, j)$ and let $(k, \ell) = (i, j+m)$. We have $i + p_{i,j} =
  (m + j - i)m$ which means that, when starting in state $i$, one can reach
  state $m$ by reading $a^{p_{i,j}}$. Being in state $m$, one of the
  $b$-transitions leads to state $m+1$. From there on, we make a single step
  backwards whenever reading the factor $a^m$. Thus, by reading the word
  $a^{q_{i,j}}$, we perform $(m+i-j+2) - i = m-j+2$ backward steps in total,
  finally reaching state $n+1 - (m-j+2) = 2m+2-(m-j+2) = m + j = \ell$.
  The converse direction of our claim follows immediately since the automaton
  is deterministic when restricted to $a$-transitions and since one can only
  reach states $\ell > m$ by using the transition $(m, b, m+1)$.

  For $X \subseteq M \times M$, we now define $u_X$ as the concatenation of
  all $u_{i,j}$ with $(i,j) \in X$, where the factors are ordered according to
  their indices $(i, j)$. By the above argument, it is easy to see that there
  is a path from state $i$ to $j + m$ labeled by $u_X$ if and only if $(i, j)
  \in X$.
  Since there are $2^{m^2} = 2^{{(n-1)}^2/4}$ subsets of the Cartesian product
  $M \times M$, it remains to show that for different subsets $X, Y \subseteq M
  \times M$, we have $u_X \not\equiv_L v_Y$.
  To this end, assume without loss of generality that $(i, j) \in X \setminus
  Y$. Then $a^{i-1} u_X a^{n-j} \in L$ but $a^{i-1} u_Y a^{n-j} \not\in L$, as
  desired.

  For the Büchi case note that for all $i \in Q$, we have $(i, b, n) \in
  \delta$ if and only if $i = n$. Therefore, by
  Proposition~\ref{prop:transfer} and the arguments above, the smallest
  semigroup weakly recognizing $\LBA(\autA)$ has at least $2^{{(n-1)}^2/4}$
  elements.
\end{proof}

The construction above does not reach the $2^{n^2}$ bound obtained when using a
larger alphabet. However, this is not surprising, given the following result.
\begin{proposition}
  Let $m \in \mathbb{N}$ be a fixed integer and let $A$ be an alphabet of size
  $m$.
  Then there exists an integer $n_m \ge 1$ such that for each finite automaton
  $\autA$ over $A$ with $n \ge n_m$ states, the language $\LNFA(\autA)
  \subseteq A^*$ (resp.\ $\LBA(\autA) \subseteq A^\omega$) is recognized
  (resp.\ weakly recognized) by a morphism onto a semigroup with less than
  $2^{n^2}$ elements.
\end{proposition}
We do not give a full proof of the proposition here, but the claim essentially
follows from a careful analysis of the subsemigroup of the transition semigroup
generated by the transitions corresponding to the letters in $A$. Applying
Devadze's Theorem~\cite{Devadze68danbssr,Konieczny11sf} to the matrix
representation of this subsemigroup shows that it is proper, \ie{}smaller than
the full transition semigroup itself.

\section{From Weakly Recognizing Morphisms to Automata}
\label{sec:hom-aut}

The well-known construction to convert weakly recognizing morphisms to finite
automata with a Büchi-acceptance condition has quadratic blow-up~\cite{pp04:short}.
We show that this is optimal up to a constant factor.

\begin{theorem}
  Let $A = \os{a, b}$, let $n \ge 3$, and let $L = \bigunion_{i=1}^n {(b a^i b
  A^*)}^\omega$.
  Then there exists a semigroup with $4n+3$ elements that weakly recognizes $L$
  and every finite automaton Büchi-accepting $L$ has at least $n(n+1) / 2$
  states.
\end{theorem}
\begin{proof}
  We first define a semigroup $S = \set{a^i, a^i b, b a^i, b a^i b}{1 \le i \le
  n} \union \os{b, bb, 0}$ by the multiplication $0 \cdot s = s \cdot 0 = 0$
  for all $s \in S$ and
  \begin{align*}
    b^\ell a^i b^r \cdot b^m a^j b^s & =
    \begin{cases}
      bb & \text{if~} i = j = 0 \\
      b^\ell a^{i+j} b^s & \text{if~} r = m = 0 \text{~and~} 1 \le i + j \le n \\
      0 & \text{if~} r = m = 0 \text{~and~} i + j > n \\
      b^\ell a^i b & \text{otherwise}
    \end{cases}
  \end{align*}
  where $\ell, m, r, s \in \os{0, 1}$ and $i, j \in \os{0, \dots, n}$.
  The morphism $h \colon A^+ \to S$ defined by $h(a) = a$ and $h(b) = b$ now
  weakly recognizes $L$ since $L$ is the union of all sets $[b a^i b]{[b a^i
  b]}^\omega$ with $1 \le i \le n$.

  Now assume that we are given a finite automaton $\autA = (Q, A, \delta, I,
  F)$ such that $\LBA(\autA) = L$.
  For each $i \in \os{1, \dots, n}$, we consider the word $\alpha_i = {(b a^i
  b)}^\omega$ and let $r_i$ be an accepting run of $\alpha_i$.
  We first show that for $i \ne j$, we have $\inf(r_i) \intersect Q \intersect
  \inf(r_j) = \emptyset$, and then prove that $\abs{\inf(r_i) \intersect Q} \ge
  i$ for $1 \le i \le n$.
  Together, this yields
  \begin{equation*}
    \abs{Q} \ge \sum_{i=1}^n \abs{\inf(r_i) \intersect Q} \ge \sum_{i=1}^n i = n(n+1) / 2.
  \end{equation*}

  Let $i, j \in \os{1, \dots, n}$ such that $i \ne j$.
  We assume for the sake of contradiction that there exists a state $q \in Q$
  with $q \in \inf(r_i)$ and $q \in \inf(r_j)$. Let $u \in b a^i b A^*$ be a
  prefix of $\alpha_i$ such that $r_i$ visits $q$ after reading $u$. Let $v \in
  A^*$ be a factor of $\alpha_j$ such that there exists a finite run labeled by
  $v$, which starts and ends in $q$, visits at least one final state and such
  that $v^\omega = {(b a^j b)}^\omega$ or $v^\omega = a^k b {(b a^j b)}^\omega$
  for some $k \in \os{0, \dots, j}$. Obviously, we then have $uv^\omega \in
  \LBA(\autA)$ but $uv^\omega \not\in L$, a contradiction.

  For the second part of the proof, assume again for the sake of contradiction
  that $\abs{\inf(r_i) \intersect Q} < i$ for some accepting run $r_i$ of
  $\alpha_i$. Then inside each $b a^i b$-factor, a state is visited twice and
  we can apply the standard pumping argument to show that a word in $A^\omega
  \setminus \LBA(\autA)$ has an accepting run as well.
\end{proof}

\section{Complementation}

To date, the best construction for complementing weakly recognizing morphisms
is the so-called \emph{strong expansion}~\cite{pp04:short}. Given a morphism $h
\colon A^+ \to S$, the strong expansion of $h$ is a morphism $g \colon A^+ \to
T$ which strongly recognizes all languages weakly recognized by $h$. If $S$ has
$n$ elements, the size of $T$ is $2^{n^2}$.
The purpose of this section is to give a lower bound for complementation.
At the same time, the established bound also serves as a lower bound for the
conversion of weak recognition to strong recognition since each morphism
strongly recognizing a language also strongly recognizes its complement.

Complementing weakly recognizing morphisms is easy in the case of $\gJ$-trivial
semigroups since each language weakly recognized by a $\gJ$-trivial semigroup
$S$ is already strongly recognized by $S$, \ie{}there is no need the compute the
strong expansion if the $\gJ$-classes of the input are trivial already. In
order to establish a lower bound, we thus consider the class of simple
semigroups, which is dual to $\gJ$-trivial semigroups in the sense that simple
semigroups consist of a single $\Jeq$-class only.

\begin{proposition}
  Let $n \ge 1$ be an arbitrary integer and let $A = \os{a_1, a_2, \dots,
  a_n}$.
  The language $L = \bigcup_{i=1}^n {(a_i A^*)}^\omega$ is weakly recognized by
  a simple semigroup with $n$ elements and every semigroup weakly recognizing
  $A^\omega \setminus L$ has at least $n \ms 2^{n-1}$ elements.
\end{proposition}

\begin{proof}
  The alphabet $A$ can be extended to a semigroup by defining an associative
  operation $a \circ b = a$ for all $a, b \in A$. Now, the morphism $h \colon
  A^+ \to (A, \circ)$ given by $h(a) = a$ for all $a \in A$ weakly recognizes
  $L$. The semigroup $(A, \circ)$ contains $\abs{A} = n$ elements and it is
  simple because we have $a \Leq b$ for all $a, b \in A$.

  Now, let $h \colon A^+ \to S$ be a morphism weakly recognizing $A^\omega
  \setminus L$.  For a letter $b \in A$ and a subset $B \subseteq A \setminus
  \os{b}$, let $u_{b,B}$ be the uniquely defined word $b a_{i_1} a_{i_2} \cdots
  a_{i_\ell}$ such that $i_1 < i_2 < \cdots < i_\ell$ and $\os{a_{i_1},
  a_{i_2}, \dots, a_{i_\ell}} = B$.
  Consider two letters $b, c \in A$ and subsets $B \subseteq A \setminus
  \os{b}$, $C \subseteq A \setminus \os{c}$. If $b \ne c$, we have $u_{b,B}
  c^\omega \not\in L$ and $u_{c,C} c^\omega \in L$.
  If $B \ne C$ we may assume, without loss of generality, that there exists a
  letter $a \in B \setminus C$. In this case, we have $a {u_{c,C}}^\omega
  \not\in L$ but $a {(u_{b,B} u_{c,C})}^\omega \in L$ and $a {(u_{c,C}
  u_{b,B})}^\omega \in L$.
  By Theorem~\ref{thm:lower}, this suffices to conclude that $h(u_{b,B}) \ne
  h(u_{c,C})$ whenever $b \ne c$ or $B \ne C$ and therefore, $S$ contains at
  least $\abs{A} \ms 2^{\abs{A}-1} = n \ms 2^{n-1}$ elements.
\end{proof}

Rather surprisingly, the established lower bound turns out to be asymptotically
tight in the case of simple semigroups. More generally, for simple semigroups,
the construction of the strong expansion can be improved such that only $n \ms
2^n$ elements are needed.
This will be proved in the remainder of this section.

We start with a morphism $h \colon A^+ \to S$ onto a simple semigroup with $n =
\abs{S}$ elements.
Since $S$ is simple, there exists a surjective mapping $\gamma \colon S \to G$
onto a finite group $G$ that becomes a bijection when restricted to a single
$\gH$-class.
Therefore, the mapping $\pi \colon ({S} / {\gR}) \times G \times
({S} / {\gL}) \to S$ with $\pi^{-1}(s) = (R_s, \gamma(s), L_s)$ for all
$s \in S$ is well-defined and bijective.
Moreover, for $s, t \in S$, we write $R_t \cdot s$ to denote the element
$\pi(R_t, \gamma(s), L_s)$.

Let $T = \set{(s, X)}{s \in S, X \subseteq S}$ and let $g \colon A^+ \to T$ be
defined by
\begin{equation*}
  g(u) = (h(u), \set{R_{h(q)} \cdot h(p)}{p, q \in A^+, pq = u})
\end{equation*}
for all $u \in A^+$. The set $T$ can be extended to a semigroup by defining an
associative multiplication
\begin{equation*}
  (s, X) \cdot (t, Y) = (st, X \union \os{R_t \cdot s} \union \hat{Y})
\end{equation*}
where $\hat{Y}$ denotes the set $\set{\pi(R_y, \gamma(s (R_t \cdot y)), L_y)}{y
\in Y}$. Under this extension, the mapping $g$ becomes a morphism.

The following three technical lemmas capture important properties of the
construction and are needed for the main proof.

\begin{lemma}\label{lem:Rts}
  Let $s, t \in S$. Then $R_t \cdot s$ is the unique element $x$ such that $x
  \Req t$, $x \Leq s$ and $\gamma(x) = \gamma(s)$ or, equivalently, the unique
  element $x$ such that $x \Heq ts$ and $\gamma(x) = \gamma(s)$.
\end{lemma}
\begin{proof}
  Let $x = R_t \cdot s$. We have $(R_x, \gamma(x), L_x) = \pi^{-1}(x) =
  \pi^{-1}(R_t \cdot s) = (R_t, \gamma(s), L_s)$. Together with the fact that
  $\pi$ is bijective, this establishes the first claim. For the second claim,
  note that since $S$ is simple, $x \Req t$ is equivalent to $x \Req ts$ and $x
  \Leq s$ is equivalent to $x \Leq ts$.
\end{proof}

\begin{lemma}\label{lem:fact}
  Let $u \in A^+$ with $g(u) = (s, X)$ and let $x \in S$. Then $x \in X \union
  \os{s}$ if and only if there exists a factorization $u = pq$ with $p \in A^+$
  and $q \in A^*$ such that $x \Heq h(qp)$ and $\gamma(x) = \gamma(h(p))$.
\end{lemma}
\begin{proof}
  Obviously, we have $x = s$ if and only if there exists a factorization $u =
  pq$ with $p = u$ and $q = \varepsilon$ satisfying the properties described
  above. Thus, it suffices to consider factorizations where $p, q \in A^+$.
  By Lemma~\ref{lem:Rts}, such a factorization exists if and only if $x =
  R_{h(q)} \cdot h(p)$ which is, in turn, equivalent to $x \in X$ by the
  definition of $g$.
\end{proof}

\begin{lemma}\label{lem:complement-correctness}
  Let $(t, f)$ be a linked pair of $S$, let $\big((s, X), (e, Y)\big)$ be a
  linked pair of $T$ and let $\alpha \in [(s, X)]_g {[(e, Y)]}_g^\omega$.
  Then $\alpha \in [t]_h{[f]}_h^\omega$ if and only if $tq = s$, $pq = e$, $qp
  = f$, $R_q \cdot t \in X$ and $R_q \cdot p \in Y$ for some $p, q \in S$.
\end{lemma}
\begin{proof}
  For the direction from left to right, let $\alpha = u v_1 v_1' v_2 v_2'
  \cdots$ such that $g(u) = (s, X)$, $g(v_i v_i') = (e, Y)$, $h(u v_1) = t$ and
  $h(v_i' v_{i+1}) = f$ for all $i \ge 1$. Furthermore, we assume without loss
  of generality that $v_i, v_i' \ne \varepsilon$ for all $i \ge 1$ and that
  $h(v_1) = h(v_2)$. We set $p = h(v_1) = h(v_2)$ and $q = h(v_1')$. Now, $tq =
  h(u v_1 v_1') = se = s$, $pq = h(v_1 v_1') = e$ and $qp = h(v_1' v_2) = f$.
  Moreover, by the definition of $g$, we have $R_q \cdot t = R_{h(v_1')} \cdot
  h(u v_1) \in X$ and $R_q \cdot p = R_{h(v_1')} \cdot h(v_1) \in Y$.

  For the converse implication, note that by Lemma~\ref{lem:fact}, there exists
  a factorization $\alpha = u v_1 v_1' v_2 v_2' \cdots$ such that $h(u) = s$,
  $h(v_i v_i') = e$, $R_{h(v_1')} \cdot h(u v_1) = R_q \cdot t$ and
  $R_{h(v_i')} \cdot h(v_i) = R_q \cdot p$ for all $i \ge 1$.
  Since $S$ is simple, $h(v_i) \Req h(v_i v_i') = e \Req p$ and $h(v_i) \Leq
  (R_{h(v_i')} \cdot h(v_i)) = (R_q \cdot p) \Leq p$ for all $i \ge 1$.
  Furthermore, $\gamma(h(v_i)) = \gamma(R_{h(v_i')} \cdot h(v_i)) =
  \gamma(R_q \cdot p) = \gamma(p)$. Together, this yields $h(v_i) = p$ by
  Lemma~\ref{lem:Rts}.
  Similarly, we have $h(v_i') \Req (R_{h(v_i')} \cdot h(v_i)) = (R_q \cdot
  p) \Req q$ and thus, $p \ms h(v_i') = h(v_i v_i') = pq$ implies $h(v_i') = q$
  for all $i \ge 1$ by Lemma~\ref{lem:inverse}.
  This shows that $h(u v_1) = sp = tqp = tf = t$ and $h(v_i' v_{i+1}) = qp =
  f$. We conclude that $\alpha \in [t]{[f]}^\omega$.
\end{proof}

\begin{theorem}\label{thm:complement-simple}
  Let $h \colon A^+ \to S$ be a morphism onto a simple semigroup of size $n =
  \abs{S}$ that weakly recognizes a language $L \subseteq A^\omega$. Then there
  exists a morphism $g \colon A^+ \to T$ to a semigroup of size $\abs{T} = n
  \ms 2^n$ that strongly recognizes $L$.
\end{theorem}

\begin{proof}
  The construction we use is the one described in the introduction of this
  section.
  Consider a linked pair $((s, X), (e, Y))$ of $T$ as well as two infinite
  words $\alpha, \beta \in [(s, X)]{[(e, Y)]}^\omega$. If $\alpha \in L$, there
  exists a linked pair $(t, f)$ of $S$ such that $\alpha \in [t]{[f]}^\omega
  \subseteq L$. Lemma~\ref{lem:complement-correctness} immediately yields
  $\beta \in [t]{[f]}^\omega \subseteq L$, thereby showing that $g$ strongly
  recognizes $L$.
\end{proof}

\section{Discussion and Open Problems}

We presented lower bound techniques and gave tight bounds for the conversion
between finite automata and weakly recognizing morphisms.
One can use techniques similar to those described in Section~\ref{sec:aut-hom}
to obtain a $3^{n^2}$ lower bound for the conversion of finite automata with
transition-based Büchi acceptance to strongly recognizing morphisms. However,
with the usual state-based Büchi acceptance criterion, the analysis becomes
much more involved and it is not clear whether the $3^{n^2}$ upper bound can be
reached. Analogously, there is no straightforward adaptation of the conversion
of weakly recognizing morphisms into B\"uchi automata in
Section~\ref{sec:hom-aut} to strongly recognizing morphisms. It would be
interesting to see whether the quadratic lower bound also holds in this
setting.

Another open problem is to close the remaining gaps between the upper and the
lower bounds. This is particularly true for the complexity of complementation
and the conversion of weakly recognizing morphisms to strong recognition. We
showed that there is an exponential lower bound and gave an asymptotically
optimal construction for simple semigroups which was a first candidate for
semigroups that are hard to complement. It is easy to adapt this construction
to families of semigroups where the size of each $\gJ$-class is bounded by a
constant. However, for the general case, the gap between $n \ms 2^{n-1}$ and
$2^{n^2}$ remains.

Beyond that, another direction for future research is to investigate whether
any of the bounds can be improved by considering the size of the accepting set,
\ie{}the number of linked pairs used to describe a language.

\bibliographystyle{abbrv}

\end{document}